\documentclass{article}

\usepackage{microtype}
\usepackage{graphicx}
\usepackage{subfigure}
\usepackage{booktabs} % for professional tables

\usepackage{hyperref}
\usepackage{multirow}

\usepackage[accepted]{icml2025}

\usepackage{amsmath}
\usepackage{amssymb}
\usepackage{mathtools}
\usepackage{amsthm}

\usepackage[capitalize,noabbrev]{cleveref}

\theoremstyle{plain}
\newtheorem{theorem}{Theorem}[section]

\theoremstyle{definition}

\newtheorem{condition}[theorem]{Condition}
\theoremstyle{remark}

\usepackage[textsize=tiny]{todonotes}

\icmltitlerunning{SMILES-to-Pharmacokinetics Diffusion with Deep Molecular Understanding}

\begin{document}

\twocolumn[
\icmltitle{Drug Discovery SMILES-to-Pharmacokinetics Diffusion Models with Deep Molecular Understanding}



\icmlsetsymbol{equal}{*}

\begin{icmlauthorlist}
\icmlauthor{Bing Hu}{cs}
\icmlauthor{Anita Layton}{math}
\icmlauthor{Helen Chen}{sphs}
\end{icmlauthorlist}

\icmlaffiliation{cs}{Computer Science, University of Waterloo, Canada}
\icmlaffiliation{math}{Applied Math, University of Waterloo, Canada}
\icmlaffiliation{sphs}{Public Health Sciences, University of Waterloo, Canada}

\icmlcorrespondingauthor{Bing Hu}{b25hu@uwaterloo.ca}

\icmlkeywords{Machine Learning, ICML}

\vskip 0.3in
]



\printAffiliationsAndNotice{} % otherwise use the standard text.

\begin{abstract}
The role of Artificial Intelligence (AI) is growing in every stage of drug development.
% Continuing breakthroughs in AI-based methods for drug discovery require the creation, improvement, and refinement of drug discovery data. 
Nevertheless, a major challenge in drug discovery AI remains: Drug pharmacokinetic (PK) datasets collected in different studies often exhibit limited overlap, creating data overlap sparsity.
Thus, data curation becomes difficult, negatively impacting downstream research investigations in high-throughput screening, polypharmacy, and drug combination. 
We propose Imagand, a novel SMILES-to-Pharmacokinetic (S2PK) diffusion model capable of generating an array of PK target properties conditioned on SMILES inputs that exhibit data overlap sparsity. 
We show that Imagand-generated synthetic PK data closely resembles real data univariate and bivariate distributions, and can adequately fill in gaps among PK datasets. As such, Imagand is a promising solution for data overlap sparsity and may improve performance for downstream drug discovery research tasks. Code available at: \url{https://github.com/GenerativeDrugDiscovery/imagand}
% We propose a novel diffusion GNN model Syngand capable of generating ligand and pharmacokinetic data end-to-end. 
% We show and provide a methodology for sampling pharmacokinetic data for existing ligands using our Syngand model. 
% We show the initial promising results on the efficacy of the Syngand-generated synthetic target property data on downstream regression tasks with AqSolDB, LD50, and hERG central. 
% Using our proposed model and methodology, researchers can easily generate synthetic ligand data to help them explore research questions that require data spanning multiple datasets.
\end{abstract}

\section{Introduction}
Generative AI is set to transform drug discovery, where it may cost \$2-3 billion dollars and 10-15 years to bring a single drug candidate to market \cite{kim2021comprehensive}. Generative AI for high-throughput screening (HTS) of ligand candidates reduces drug development costs and is changing how ligands are designed and tested \cite{pushpakom2019drug}. Initial success of drug discovery AI has been in drug repurposing \cite{thafar2022affinity2vec, Morselli_Gysi_2021}, drug-target interaction \cite{lian2021drug}, drug response prediction \cite{pouryahya2022pan}, poly-pharmacy \cite{vzitnik2015gene}, and the generation of synthetic ligands and drug properties \cite{digress, hu2024syntheticdatadiffusionmodels}. Thus far, what has advanced drug discovery AI is a continued effort towards open data for training and testing \cite{huang2021therapeutics, guacamol, gaulton2017chembl}.

Data collection for drug discovery through assay panels is expensive and time-consuming.
% Although there are clear advances toward standardizing clinical, chemical, and pre-clinical data and making such datasets openly available \cite{kim2023pubchem, huang2021therapeutics, nusinow2020quantitative}, the lack of overlap among datasets remains a challenge when one tries to combine data \cite{scoartaroadmap}. 
Although there are clear advances toward standardization and dissemination of pre-clinical, clinical, and chemical datasets \cite{kim2023pubchem, huang2021therapeutics, nusinow2020quantitative}, challenges arise when merging and linking these datasets together \cite{scoartaroadmap}.
Collected independently, drug discovery datasets often have limited overlap, which poses a challenge for researchers looking to answer research questions requiring data from multiple datasets. One notable example is the study of drug combinations and polypharmacy \cite{scoartaroadmap}.

Recent advances in drug discovery AI have utilized Denoising Diffusion Probabilistic Models (DDPMs) \cite{ddpm}, which yield a new class of diffusion models capable of generating ligand structures \cite{guo2023diffusion, digress, diffbridge, difflinker}. \citet{hu2024syntheticdatadiffusionmodels} have shown that diffusion models can generate pharmacokinetic (PK) properties alongside the ligand diffusion pipeline with promising results. 
% Multi-modal diffusion models such as Text-to-Image and Text-to-Video diffusion models have been demonstrated to generate high-quality photorealistic data \cite{saharia2022photorealistic, ho2022videodiffusionmodels}. 
% \citet{azizi2023synthetic} have also shown that synthetic data from diffusion models has improved performance in ImageNet classification tasks. 
Inspired by these advances,  we propose Imagand, which can generate 12 PK target properties from 10 PK datasets conditioned on learned SMILES embeddings. 
% We show that the properties of our diffusion-generated synthetic data closely match that of the real data and improve performance on downstream drug discovery prediction tasks. 
Specifically, our contributions are as follows:
\begin{itemize}
    \item We propose Imagand, a novel multi-modal SMILES-to-Pharmacokinetic (S2PK) diffusion model capable of generating an array of target properties conditioned on learned SMILES embeddings. 
    \item We develop a noise model that creates a prior distribution closer to the true data distribution, improving performance.
    \item We show that synthetic data generated from our Imagand model has univariate and bivariate distributions closely matching the real data and improves machine learning efficiency. 
\end{itemize}
Notably, Imagand generates dense synthetic data that overcomes the challenges of sparse PK datasets with limited overlap. Using Imagand, researchers can generate large synthetic PK assays over thousands of ligands to answer poly-pharmacy and drug combination research questions at a fraction of the cost of conducting \textit{in vitro} or \textit{in vivo} PK assay panels. 

\section{Background}

Diffusion methods use families of probability distributions to model complex datasets for computationally tractable learning, sampling, inference and evaluation \cite{guo2023diffusion}. Denoising Diffusion Probabilistic Models (DDPM) \cite{ddpm} first systematically destroy the structure in the data through a forward process, and then in a reverse process, learn how to restore the structure in the data from noise. Recent literature has covered many advances in small-molecule generation using diffusion models \cite{cdgs, edm, egnn, digress}. 

% DDPMs can be combined with graph networks. \textbf{C}onditional \textbf{D}iffusion models are based on discrete \textbf{G}raph \textbf{S}tructures (CDGS) and can be used to generate small-molecule graphs with similar distribution to real small-molecules \cite{cdgs}. 
% % As molecules can be equivariant to rotations, translations, and permutations, E(n)-equivariant GNN (EGNN) \cite{egnn} is a model that leverages this equivariance. ENGG is also the basis for E(3)-equivariant diffusion models (EDM) \cite{edm} that perform the diffusion process in Euclidean space. 
% Digress \cite{digress} combines graph transformers \cite{dwivedi2020generalization} and discrete diffusion \cite{austin2021structured} for molecular-conditioned small-molecule generation. Digress utilizes graph-level properties such as cycles, and spectral and molecular features to augment the input, improving training and sampling performance. 
% Conditional generation has shown benefits for diffusion Text-to-Image and Text-to-Video models \cite{saharia2022photorealistic, ho2022videodiffusionmodels}. In Text-to-Image and Text-to-Video models, learned embeddings from Large Language Models (LLMs) are utilized as input to lend deep language understanding to the diffusion models. 

PK broadly describes what the body does to a drug regarding absorption (how the body absorbs the drug), bioavailability (the extent the active drug enters circulation), distribution (how the drug distributes in tissue), metabolism (how the body breaks down the drug), and excretion (how the drug is removed from the body). 
% Issues related to PK properties are a primary driver for compound attrition in terms of small-molecule drug development \cite{kola}, although improving in recent times with advances in PK computational tools \cite{Waring2015AnAO, DAVIES2020390}. 
As issues related to PK properties are the primary drivers for compound attrition for small-molecule drug development \cite{kola}, accurate PK computational tools are critical and have advanced in recent times \cite{Waring2015AnAO, DAVIES2020390, ahmed2021impact}.
Physiologically-based pharmacokinetics (PBPK) offers the modelling of PK properties using mathematical equations representing the human body \cite{Sager1823}.
% Physiologically-based pharmacokinetics (PBPK) modelling offers a mathematical framework to simulate the time course of a compound and its ADME properties [
% 9
% ]. PBPK can be used to understand in vivo behaviour and to extrapolate to humans and is normally applied in the later stages of the drug discovery process. All these methods rely on expensive in vitro and in vivo animal experiments and cannot be utilised in a high-throughput manner on a large number of compounds.
PBPK rely on expensive \textit{in-vitro} and \textit{in-vivo} human and animal experiments and cannot be utilized in high-throughput screening across large numbers of ligands (10K to 100K drugs per day) \cite{obrezanova2023artificial}. 

% and is expensive and computationally intensive to scale to support high-throughput screening across large numbers of drugs (10K to 100K ligands per day) \cite{obrezanova2023artificial}.

Extending many PK properties across large arrays of ligands can be costly given the expense associated with data collection for drugs. 
Consequently, oftentimes only small sets of ligands can be feasibly tested for target property data collection studies, leading to minimal overlap between collected datasets \cite{scoartaroadmap}.
Comparing the 11 PK datasets we use in this study, \autoref{tab:totalsparse} shows the minimal overlap sparsity between all of the datasets. 
This challenge poses barriers for scientists interested in answering research questions requiring data across multiple datasets, such as in poly-pharmacy and drug combination research. 

\section{Methodology}

The choice of noise models for noising in diffusion models may have a substantial impact on performance; using a prior distribution close to the true data distribution can make training easier \cite{digress}. As PK properties do not always follow a Gaussian or uniform noise model, we propose a noise model called Discrete Local Gaussian Noise (DLGN). 

\subsection{Discrete Local Gaussian Sampling}

Discrete local gaussian (DLG) sampling is based on principles of inverse transform sampling. Inverse transform sampling is a method for sampling from any probability distribution given its cumulative distribution function. For any variable $X\in \mathbb{R}$, the random variable $F_X^{-1}(U)$ has the same distribution as $X$, where $F_X^{-1}$ is the generalized inverse of the cumulative distribution function $F_X$ of $X$ and $U \sim Unif[0,1]$. 

As the true distribution is not always available, Discrete Local Gaussian Sampling $\mathbb{DLGS}$ looks to discretely approximate the cumulative distribution function by combining binning and Gaussian noise given by \autoref{th:kl1} (with proof \autoref{proof:kl}). Given $N$ bins, a discrete cumulative distribution function $\hat{F}_{X_N}^{-1}$ can be constructed. With $\sigma$ as a scaling factor, we can then define DLG sampling $\phi(X=x)$ as:
\begin{align}
    \mathbb{\phi}(X=x) &:= \mathcal{N}(x; \hat{F}_{X_N}^{-1}(U), \frac{\sigma}{N^2}\mathbf{I})
\end{align}
\begin{theorem}
\label{th:kl1}
Let $D_{KL}(P||Q)$ be the Kullback-Leibler (KL) divergence for a model probability distribution $Q$ and the true probability distribution $P$. For any distribution $X$, we then have:
\[lim_{N \to \infty}\ D_{KL}(X \ ||\  \mathbb{\phi}(X=x)) = 0\]        
\end{theorem}
\begin{figure}[t]
\centering
\includegraphics[width=0.9\columnwidth]{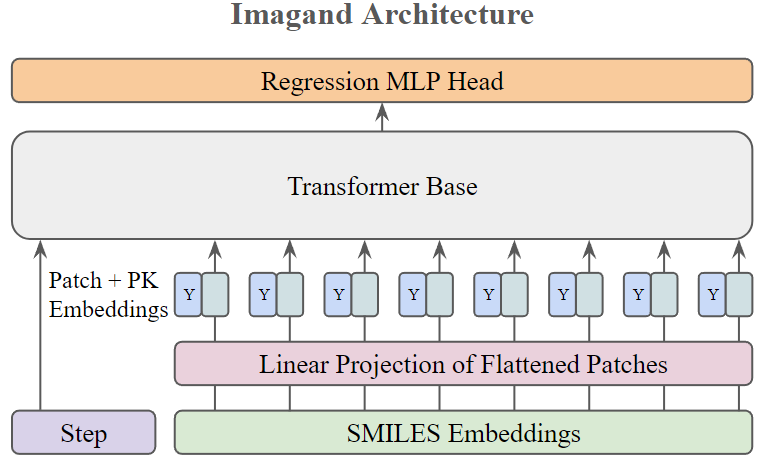} % Reduce the figure size so that it is slightly narrower than the column. Don't use precise values for figure width.This setup will avoid overfull boxes.
\caption{Model overview. Patches are generated from SMILES Embeddings combined with PK embeddings. The patches are then fed along with step embeddings into the base transformer model. A regression MLP head is used to produce the necessary output for denoising.}
\label{fig:arch}
\end{figure}

% of the SMILES-to-Pharmacokinetic model 

% As ChemBERTa is trained on a wide array of SMILES, embeddings from ChemBERTa are an effective way to inject deep molecular understanding into our diffusion model. Similar to \citet{saharia2022photorealistic}, we freeze the weights of ChemBERTa. Freezing of weights provides benefits in negligible computation and memory footprint during training of the SMILES-to-Pharmacokinetic model due to the offline computation of embeddings. 

\subsection{Imagand Model}

Imagand is an S2PK diffusion model conditioned on learned SMILES embeddings from SMILES encoder models to generate target PK properties. Imagand improves accuracy by training a diffusion process \cite{ddpm} utilizing custom noise models.

\subsubsection{Diffusion Model}
Given samples from a data distribution $q(x_0)$, we are interested in learning a model distribution $p_\theta(x_0)$ that approximates $q(x_0)$ and is easy to sample from. \citet{ddpm} considers the following Markov chain with Gaussian transitions parameterized by a decreasing sequence $\alpha_{1:T}\in(0,1]^T$:
\begin{align}
    q(x_{1:T}|x_0) := \mathcal{N}(x_{1:T}|\sqrt{\alpha_{1:T}}x_0, (1-\alpha_{1:T})\mathbf{I})
\end{align}
This is called the forward process, whereas the latent variable model $p_\theta(x_{0:T})$ is the generative process, approximating the \textit{reverse process} $q(x_{t-1}|x_t)$. The forward process of $x_t$ can be expressed as a linear combination of $x_0$ and noise variable $\epsilon$:
\begin{align}
    x_t = \sqrt{\alpha_{t}}x_0\ +\ \sqrt{1-\alpha_{t}} \epsilon
\end{align}
We experiment with noise models $\epsilon$, testing with Average, Uniform, Gaussian, and DLG sampling approaches. Empirically, from our ablation studies, we find that DLG sampling improves performance over other noise models, as it better approximates the original data. We train with the simplified objective:
% \begin{align}
%     L(\epsilon_\theta) := \sum_{t=1}^T\mathbb{E}_{x_0\sim q(x_0), \epsilon_t}[||\epsilon_\theta^{(t)}(\sqrt{\alpha_t}x_0+\sqrt{1-\alpha_t}\epsilon_t) - \epsilon_t ||_2^2]
% \end{align}
\begin{align}
    L(\epsilon_\theta) := \sum_{t=1}^T\mathbb{E}_{x_0\sim q(x_0), \epsilon_t}[||\epsilon_\theta^{(t)}(x_t) - \epsilon_t ||_2^2]
\end{align}
Where $\epsilon_\theta:=\{\epsilon_\theta^{(t)}\}_{t=1}^T$ is a set of T functions, indexed by t, each with trainable parameters $\theta^{(t)}$. Convergence analysis of utilizing DLG sampling as a noise model is provided in \autoref{appdx:convergence}. 

\subsubsection{Architecture} Imagand resembles a typical vision transformer architecture \cite{dosovitskiy2021imageworth16x16words}; see Figure \ref{fig:arch}. 1D patches are computed from the classifier-free guidance of SMILES embeddings and concatenated with PK class tokens. 
Diffusion step embeddings are generated using sinusoidal position encodings \cite{vaswani2023attentionneed}. 
Patches are then fed alongside sinusoidal step embeddings \cite{ho2021cascadeddiffusionmodelshigh} to a transformer base. 
% As the input PK data is sparse over ligand data, we utilize a masked loss function to only flow gradients from known PK values during training.
We mask out missing values when computing the loss for the model only to flow gradients and learn from non-missing PK values during training.
Exponential Moving Average (EMA) \cite{tarvainen2018meanteachersbetterrole} is applied to the base model during training to generate the final model used for sampling.

\subsection{Pre-trained SMILES Encoder}

S2PK diffusion models need powerful semantic SMILE encoders to capture the complexity of arbitrary chemical structure inputs. Given the sparsity and small size of PK datasets, encoders trained on specific SMILES-Pharmacokinetic pairs are infeasible \cite{huang2021therapeutics}. 
Many transformer-based foundational models such as ChemBERTa \cite{chithrananda2020chemberta, ahmad2022chemberta}, SMILES-BERT \cite{Wang2019SMILESBERTLS}, and MOLGPT \cite{bagal2021molgpt} have been pre-trained to deeply understand molecular and chemical structures and properties.
After pre-training, these foundational models can then be fine-tuned for various downstream molecular tasks.
Language models trained on SMILES-only corpus, significantly larger than SMILES-Pharmacokinetic data, learn a richer and wider distribution of molecular and chemical structures.  

% We utilize ChemBERTa embeddings for all our current models, although embeddings from any SMILES-only pre-trained LLM could be used instead.
We test SMILES embeddings from ChemBERTa \cite{ahmad2022chemberta}, T5 \cite{raffel2023exploringlimitstransferlearning}, and DeBERTa \cite{he2021debertadecodingenhancedbertdisentangled} trained on SMILES-only corpora. 
We further test and compare embedding performance for SMILES embedding from ChemBERTa trained either on ZINC (100K molecules) \cite{irwin2005zinc} or PubChem (10M molecules) \cite{kim2023pubchem} SMILES corpora.
All SMILES embedding models were collected through the Huggingface \cite{wolf-etal-2020-transformers} Model Hub. 
As ChemBERTa, T5, and DeBERTa are all trained on a wide array of SMILES, embeddings from these models are an effective way to inject deep molecular understanding into our diffusion model. 
Similar to \citet{saharia2022photorealistic}, we freeze the weights of our embedding models.
Because embeddings are computed offline, freezing the weights minimizes computation and memory footprint for embeddings during model training.

\subsubsection{Classifier-free Guidance} 
% Classifier guidance is a technique to improve sample quality while reducing diversity in conditional diffusion models using gradients from a pre-trained model p(c|zt) during sampling.
Classifier guidance uses gradients from a pre-trained model to improve quality while reducing diversity in conditional diffusion models during sampling \cite{dhariwal2021diffusionmodelsbeatgans}. Classifier-free guidance \cite{ho2022classifierfreediffusionguidance} is an alternative technique that avoids this pre-trained model by jointly training a diffusion model on conditional and unconditional objectives via dropping the condition (i.e. with 10\% probability). We condition all diffusion models on learned SMILES embedding and sinusoidal time embeddings using classifier-free guidance through dropout \cite{ho2022classifierfreediffusionguidance, JMLR:v15:srivastava14a}.

\begin{table*}[t]
    \centering
    \begin{tabular}{l | l l l l l l l l l l l}
         Data Overlap Cardinality & \textbf{1} & 2 & 3 & 4 & 5 & 6 & 7 & 8 & 9 & 10 & 11 \\
         \hline
         Number of Data Points & \textbf{322235} & 3598 & 404 & 1110 & 34 & 105 & 27 & 4 & 14 & 2 & 0 \\
    \end{tabular}
    \caption{Number of data points for each data overlap cardinality over 11 PK datasets. Data overlap cardinality represents the number of datasets a data point is in. The number of data points in data overlap cardinality 1 represents the number of data points only in one of the 11 PK datasets. As data overlap cardinality increases, the number of data points greatly drops, describing the data overlap sparsity phenomenon. This lack of data negatively impacts downstream research investigations in high-throughput screening, poly-pharmacy, and drug combination.
    }
    \label{tab:totalsparse}
\end{table*}

\section{Experiments}

In the following, we describe the model training details and compare our synthetic data to real data, in terms of machine learning efficiency (MLE) and univariate and bivariate statistical distributions. We then discuss ablation studies and key findings. The metrics for MLE, univariate, and bivariate evaluations are further defined in their respective subsections. 

We compare Imagand with baselines of Conditional GAN (cGAN) \cite{mirza2014conditional} and Syngand \cite{hu2024syntheticdatadiffusionmodels}. Similar to in Imagand, SMILES-embeddings from a pre-trained T5 model are used conditionally by the cGAN model to generate PK properties as output for a specific drug. 
Additional training details and baseline results are provided in \autoref{appdx:training} and \autoref{appdx:results}.

\subsection{Pharmacokinetic Datasets}

All 11 PK datasets are collected from TDCommons \cite{huang2021therapeutics}. We select PK datasets suitable for regression from the absorption, distribution, metabolism, and excretion (ADME) and Toxicity categories. 
Looking over 11 PK datasets 
% \cite{Wang2016ADMEPE, wu2018moleculenet, sorkun2019aqsoldb,mobley2014freesolv, ppbr, lombardo2016silico, obach2008trend, di2012mechanistic, ld50, du2011hergcentral} 
for target property screening, \autoref{tab:totalsparse} shows data overlap sparsity as the data overlap cardinality increases, the number of data points greatly drops. 
% We cover each selected PK dataset (in bold) briefly.

\textbf{Caco-2} \cite{Wang2016ADMEPE} is an absorption dataset containing rates of 906 drugs passing through the Caco-2 cells, approximating the rate at which the drugs permeate through the human intestinal tissue. 
\textbf{Lipophilicity} \cite{wu2018moleculenet} is an absorption dataset that measures the ability of 4,200 drugs to dissolve in a lipid (e.g. fats, oils) environment.
\textbf{AqSolDB} \cite{sorkun2019aqsoldb} is an absorption dataset that measures the ability of 9,982 drugs to dissolve in water.
\textbf{FreeSolv} \cite{mobley2014freesolv} is an absorption dataset that measures the experimental and calculated hydration-free energy of 642 drugs in water.

\textbf{Plasma Protein Binding Rate (PPBR)} \cite{ppbr} is a distribution dataset of percentages for 1,614 drugs on how they bind to plasma proteins in the blood.
\textbf{Volume of Distribution at steady state (VDss)} \cite{lombardo2016silico} is a distribution dataset that measures the degree for 1,130 drugs on their concentration in body tissue compared to their concentration in blood. 

\textbf{Half Life} \cite{obach2008trend} is an excretion dataset for 667 drugs on the duration for the concentration of the drug in the body to be reduced by half. 
\textbf{Clearance} \cite{di2012mechanistic} is an excretion dataset for around 1,050 drugs on two clearance experiment types, microsome and hepatocyte. Drug clearance is defined as the volume of plasma cleared of a drug over a specified time \cite{huang2021therapeutics}.

\textbf{Acute Toxicity (LD50)} \cite{ld50} is a toxicity dataset that measures the most conservative dose for 7,385 drugs that can lead to lethal adverse effects.
\textbf{hERG Central} \cite{du2011hergcentral} is a toxicity dataset that measure the blocking of Human ether-à-go-go related gene (hERG) for 306,893 drugs. hERG is crucial for the coordination of the heart's beating. hERG contains percentages inhibitions at $1\mu{}M$ and $10\mu{}M$.

\begin{table*}[t]
\centering
\begin{tabular}{l|l|cccc}
\hline
\multicolumn{1}{c|}{}&\multicolumn{1}{c}{}&\multicolumn{4}{c}{Model}\\
Data&Metric&\multicolumn{1}{l}{{cGAN}}&\multicolumn{1}{l}{{Sygd*}}&\multicolumn{1}{l}{Imgd}&\multicolumn{1}{l}{Real}\\
\hline
\multirow{3}{*}{Caco2}&MSE&0.165&0.276&\textbf{0.131}&0.634\\
&R2&-0.08&-3.35&\textbf{0.137}&-3.215\\
&PCC&0.338&0.302&\textbf{0.426}&0.352\\
\hline
\multirow{3}{*}{Lipo.}&MSE&\textbf{0.141}&0.313&0.150&0.167\\
&R2&\textbf{0.194}&0.126&0.138&0.04\\
&PCC&0.469&0.181&0.409&\textbf{0.499}\\
\hline
\multirow{3}{*}{AqSol}&MSE&\textbf{0.074}&0.107&0.08&0.075\\
&R2&\textbf{0.565}&0.348&0.533&0.564\\
&PCC&0.755&0.681&0.731&\textbf{0.756}\\
\hline
\multirow{3}{*}{FSolv}&MSE&0.198&0.182&\textbf{0.165}&0.624\\
&R2&-0.09&-0.023&\textbf{0.078}&-2.501\\
&PCC&0.422&\textbf{0.515}&0.391&0.383\\
\hline
\multirow{3}{*}{PPBR}&MSE&\textbf{0.26}&0.361&0.263&3.527\\
&R2&-0.082&-1.36&\textbf{-0.06}&-13.31\\
&PCC&\textbf{0.225}&0.125&0.223&0.095\\
\hline
\multirow{3}{*}{VDss}&MSE&0.209&0.307&\textbf{0.196}&0.535\\
&R2&-0.064&-0.843&\textbf{-0.015}&-1.771\\
&PCC&\textbf{0.306}&0.189&0.298&0.234\\
\hline
\end{tabular}
\hspace{1em}
\begin{tabular}{l|l|cccc}
\hline
\multicolumn{1}{c|}{}&\multicolumn{1}{c}{}&\multicolumn{4}{c}{Model}\\
Data&Metric&\multicolumn{1}{l}{{cGAN}}&\multicolumn{1}{l}{{Sygd*}}&\multicolumn{1}{l}{Imgd}&\multicolumn{1}{l}{Real}\\
\hline
\multirow{3}{*}{Half}&MSE&0.284&0.437&\textbf{0.261}&0.525\\
&R2&-0.536&-0.831&\textbf{-0.275}&-1.589\\
&PCC&0.134&0.065&0.034&\textbf{0.156}\\
\hline
\multirow{3}{*}{Cl.(H)}&MSE&\textbf{0.431}&0.563&0.433&1.863\\
&R2&\textbf{-0.153}&-1.14&-0.200&-4.24\\
&PCC&\textbf{0.144}&0.032&0.096&0.109\\
\hline
\multirow{3}{*}{Cl.(M)}&MSE&\textbf{0.203}&0.278&0.209&0.717\\
&R2&\textbf{-0.037}&-1.71&-0.043&-2.599\\
&PCC&0.25&0.189&\textbf{0.253}&0.132\\
\hline
\multirow{3}{*}{LD50}&MSE&0.103&0.111&\textbf{0.100}&0.105\\
&R2&0.252&\textbf{0.298}&0.277&0.240\\
&PCC&0.526&\textbf{0.558}&0.537&0.542\\
\hline
\multirow{3}{*}{hRG.1}&MSE&0.132&\textbf{0.121}&0.127&0.136\\
&R2&-0.135&-0.55&\textbf{-0.108}&-0.189\\
&PCC&0.035&0.060&0.062&0.062\\
\hline
\multirow{3}{*}{hRG.10}&MSE&0.134&\textbf{0.106}&0.115&0.121\\
&R2&-0.182&-0.075&\textbf{-0.023}&-0.081\\
&PCC&0.207&0.200&0.196&\textbf{0.212}\\
\hline
\end{tabular}
\caption{Comparing drug discovery Machine Learning Efficiency (MLE) regression performances between different models and with real train data. Mean Squared Error (MSE), R-Squared (R2), and Pearson Correlation Coefficient (PCC) values are averaged over 30 trials, with the best scores on the real test set bolded. *Syngand R2 and PCC results are scale-adjusted relative to Real-Real with cGAN and Imagand results.
} \label{tab:res}
\end{table*}

\subsection{Data Processing}

We first merge all 11 PK datasets to create a unified dataset containing 30K drugs over 12 unique PK columns for training and testing (90\%/10\% split) our models. Excluding the hERG dataset from which we sample 7.9K drugs, we merge the remaining 9 PK datasets for 22.1K unique drugs. We arrive at a total of 30K drugs in our unified dataset after merging the 7.9K drugs sampled from hERG into our 22.1K unique drugs from the other 9 PK datasets. 
We only sample 7.9K drugs from hERG to maintain balance in the unified dataset given the size imbalance of hERG compared to the other 9 PK datasets. 
After removing outliers ($Q1-1.5$IQR lower and $Q3+1.5$IQR upper bound), we are left with 28,397 drugs from the original 30K drugs. 
The 28,397 drug values for each of the 12 PK columns are then min-max scaled between the range of $[-1,1]$. 
Outliers are removed to ensure that Min-Max normalization does not cause unwarranted skewness in our trainset distribution, causing issues for model training.
Before infilling null values using one of the average, uniform, or Gaussian distributions, or the proposed DLGN method, we store the null masks for each drug for the masked loss function.

% All code can be found at [TBD].
% The model configuration we analyze for machine learning efficiency and univariate and bivariate distributions uses DeBERTa embeddings trained on 
% The model is trained with diffusion timesteps set to 2000. 

% \begin{figure}[h]
% \centering
% \includegraphics[width=0.5\columnwidth]{Figures/noisediff.png} % Reduce the figure size so that it is slightly narrower than the column. Don't use precise values for figure width.This setup will avoid overfull boxes.
% \caption{Model overview. Patches are generated from SMILES Embeddings combined with PK embeddings. The patches are then fed along with step embeddings into the base transformer model. A regression MLP head is used to produce the necessary output for denoising.}
% \label{fig:univariate}
% \end{figure}
Using the trained S2PK model, we generate synthetic PK target properties for 3K ligands selected from our test dataset. The generated synthetic data, containing 3K ligands with all 12 target properties, can be used to augment real data for research requiring data spanning these target properties. Given the smaller size of real target property datasets, 3K synthetic target property ligands provide meaningful augmentations to the real data.

\subsection{Machine Learning Efficiency}

Machine Learning Efficiency (MLE) is a measure that assesses the ability of the synthetic data to replicate a specific use case \cite{dankar2021fake, basri2023hyperparameter, borisov2022language}.
MLE represents the ability of the synthetic data to replace or augment real data in downstream use cases. 
To measure MLE, 
% as shown in Figure \ref{fig:mle}, 
two models are trained separately using synthetic versus real data, and then their performance, measured by Mean-Squared Error (MSE), R-Squared (R2), and Pearson Correlation Coefficient (PCC), is evaluated on real data test sets and compared.   

% \begin{figure}[t]
% \centering
% \includegraphics[width=0.9\columnwidth]{Figures/mle.png}
% \caption{Overview of the MLE evaluation framework.}
% \label{fig:mle}
% \end{figure}

For this experiment, we train Linear Regression (LR) models using ChemBERTa embeddings to predict each PK target property value. 
To prevent data leakage, we first 
divide real and synthetic data before combining them to form train and test sets, as follows.
To ensure an adequately sized test set (\textgreater300 ligands, i.e. \textgreater10\% size of our synthetic data) to evaluate our downstream models, we divide real data into segments denoted  $A_{r}$ and $B_{r}$ using a 50\%/50\% split.
To ensure a synthetic test set similar in size to real data test sets ($\sim300$ ligands), we divide synthetic data into segments denoted $A_{s}$ and $B_{s}$ using a 90\%/10\% split. 
%A 50\%/50\% split is used for the real data to ensure an adequately sized test set to evaluate . 
The real train set is defined as $A_{r}$ and the real test set is defined as $B_{r}$. The augmented train set is defined as $A_{r} \cup A_{s}$ and the augmented test set is defined as $B_{r} \cup B_{s}$.
Outliers are removed from both real and augmented train and test sets based on $Q1-1.5$IQR lower and $Q3+1.5$IQR upper bounds on the synthetic data. 

Table \ref{tab:res} shows the results of the PK regression tasks using real and synthetic augmented datasets. Results of these experiments suggest that a synthetic augmented dataset can outperform real data with statistical significance over many PK datasets. Additional tasks will be explored in future work. We see that synthetic data from both cGAN and Imagand can improve MLE over using only the real data. Imagand has similar or superior MLE performance compared to cGAN. 
\begin{figure}[h]
\centering
\includegraphics[width=0.7\columnwidth]{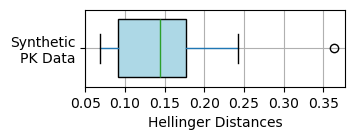}
\caption{Synthetic PK Data Hellinger Distances (HDs).}
\label{fig:hd}
\end{figure}
\subsection{Univariate Distributions}
\begin{figure}[t]
\centering
\includegraphics[width=1\columnwidth]{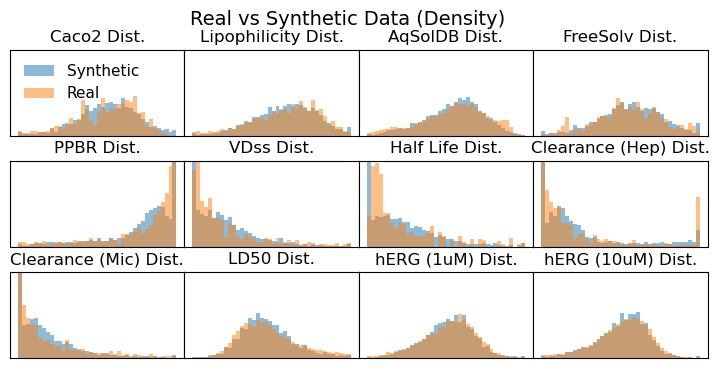} % Reduce the figure size so that it is slightly narrower than the column. Don't use precise values for figure width.This setup will avoid overfull boxes.
\caption{Distributions of ligand PK properties. Blue, synthetic distributions; orange, real distributions. 
% Log scale is used for the y-axis, in arbitrary unit.
}\label{fig:dist}
\label{fig:univariate}
\end{figure}
The generated synthetic data closely matches that of the real data; see Figure \ref{fig:dist}. 
Hellinger distance (HD) quantifies the similarity between two probability distributions and can be used as a summary statistic of differences for each PK target property between real and synthetic datasets. Given two discrete probability distributions $P = \{p_1,p_2,...,p_n\}$ and $Q = \{q_1, q_2,...,q_n\}$, the HD between $P$ and $Q$ is expressed in Equation \ref{eq:hd}. 
\begin{equation}\label{eq:hd}
HD^2(p,q) = 
\frac{1}{2} \sum^n_{i=1}\left( \sqrt{p_i} - \sqrt{q_i} \right)
\end{equation}
With scores ranging between 0 to 1, HD values closer to 0 indicate smaller differences between real and synthetic data and are thus desirable. 
Figure \ref{fig:hd} shows the HD values for our synthetic data compared to real data with the average HD being 0.15. 
% The average HD score for our generated synthetic data is 0.15.
% , with individual synthetic PK target property HD scores ranging from 0.07 to 0.36. 
\begin{table}[t]
\centering
\begin{tabular}{l|cc|cc}
\hline
\multicolumn{1}{c}{} & \multicolumn{2}{|c}{Mean} & \multicolumn{2}{|c}{Std} \\
Data                     & Real               & Syn & Real & Syn \\ \hline
\multirow{1}{*}{Caco2}  & 0.118             & 0.137        & 0.388                            & 0.375                         \\
\multirow{1}{*}{Lipophilicity}    & 0.184           & 0.179       & 0.417                              & 0.386                      \\
\multirow{1}{*}{AqSolDB}    & 0.106           & 0.107       & 0.412                               & 0.362                        \\
\multirow{1}{*}{FreeSolv}    & 0.103           & 0.123       & 0.421                              & 0.400                         \\
\multirow{1}{*}{PPBR}    & 0.570            & 0.562      & 0.496                                & 0.467                         \\
\multirow{1}{*}{VDss}    & -0.603           & -0.615       & 0.442                                & 0.389                         \\
\multirow{1}{*}{Half life}    & -0.557            & -0.559       & 0.450                               & 0.419                         \\
\multirow{1}{*}{Clearance (H)}    & -0.549        & -0.559          & 0.605                                & 0.551                         \\
\multirow{1}{*}{Clearance (M)}    & -0.670        & -0.676          & 0.445                                & 0.382                         \\
\multirow{1}{*}{LD50}    & -0.038           & -0.054       & 0.372                                & 0.331                         \\
\multirow{1}{*}{hERG 1uM}    & 0.036           & 0.031       & 0.338                                & 0.319                         \\
\multirow{1}{*}{hERG 10uM}    & 0.027           & 0.030        & 0.335                               & 0.316                         \\\hline
\end{tabular}
\caption{Comparing mean and standard deviation values between real and synthetic target property values.} \label{tab:mean}
\end{table}
Table \ref{tab:mean} compares the mean and standard deviation of the real and synthetic target property values. 
The mean and standard deviation of the generated synthetic data closely resemble that of the real data for each PK target property. 
We found that normalization combined with static thresholding substantially limits the generation of invalid and out-of-range PK values. 
Given the underpinnings of diffusion using Gaussian reparameterization \cite{ddpm}, diffusion methods have challenges learning and generating non-Gaussian data. This failure mode is evident in \autoref{fig:realvssyn} for non-Gaussian distributions PPBR, VDss, Half Life, and Clearance (Hep and Mic) where we see that the synthetic data fails to replicate the Log-logistic real-data distribution common in drug discovery datasets. 
% Our choice of DLGN noise strategy equally plays a large role in determining the quality of the generated synthetic PK target property data (more in the Ablation Studies Section).
% Computing HD scores for each synthetic PK target property, ranging from 0.07 to 0.36, the average HD score is 0.15. 

% HD provides a summary statistic of differences between each PK target prop in the real and synthetic datasets. HD scores range between 0 to 1, where values closer to 0 are desired as they indicate lower differences in the distribution between real and synthetic datasets
\begin{figure*}[h]
\centering
\includegraphics[width=1.4\columnwidth]{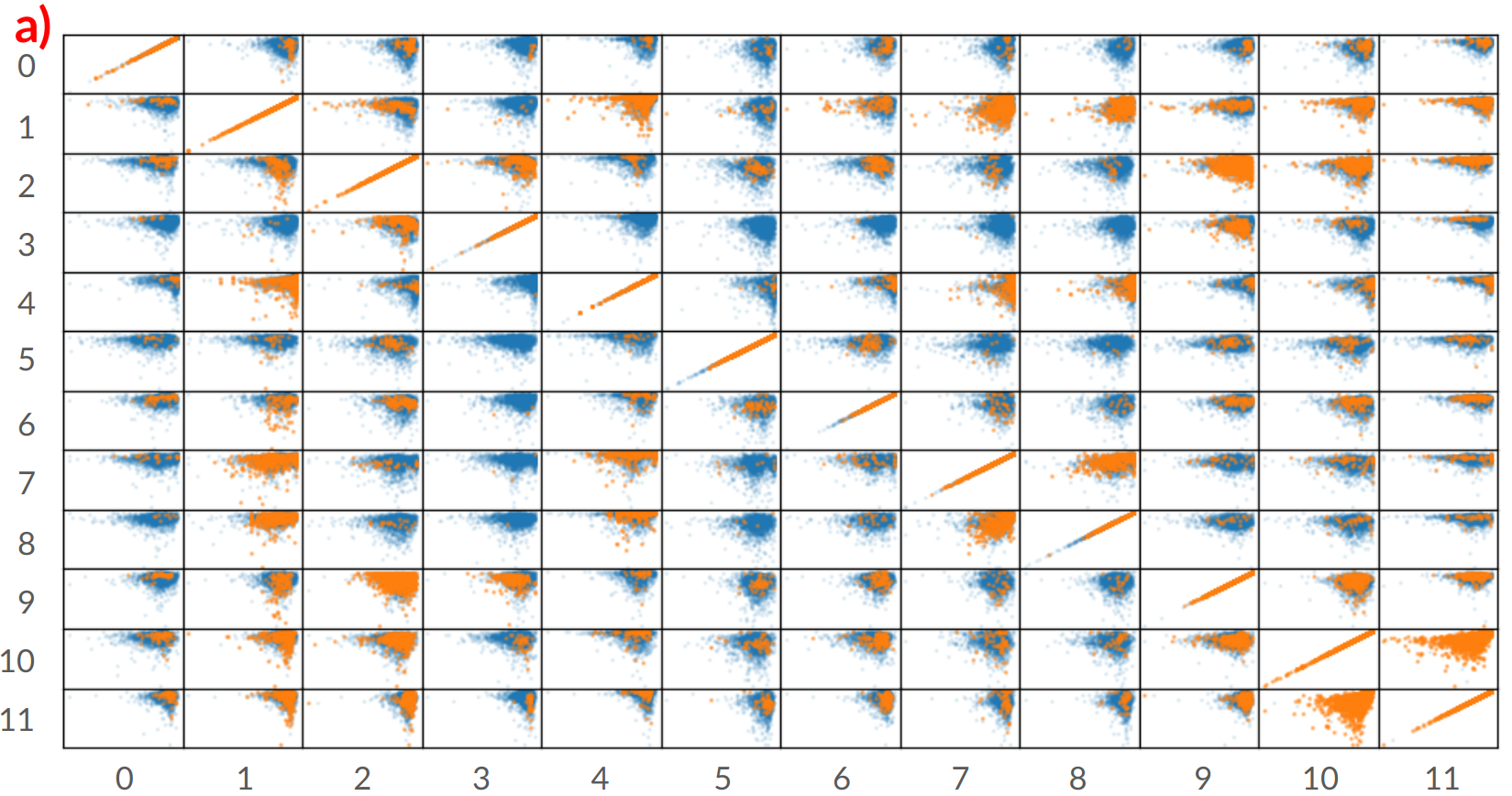} % Reduce the figure size so that it is slightly narrower than the column. Don't use precise values for figure width.This setup will avoid overfull boxes.
\includegraphics[width=1.4\columnwidth]{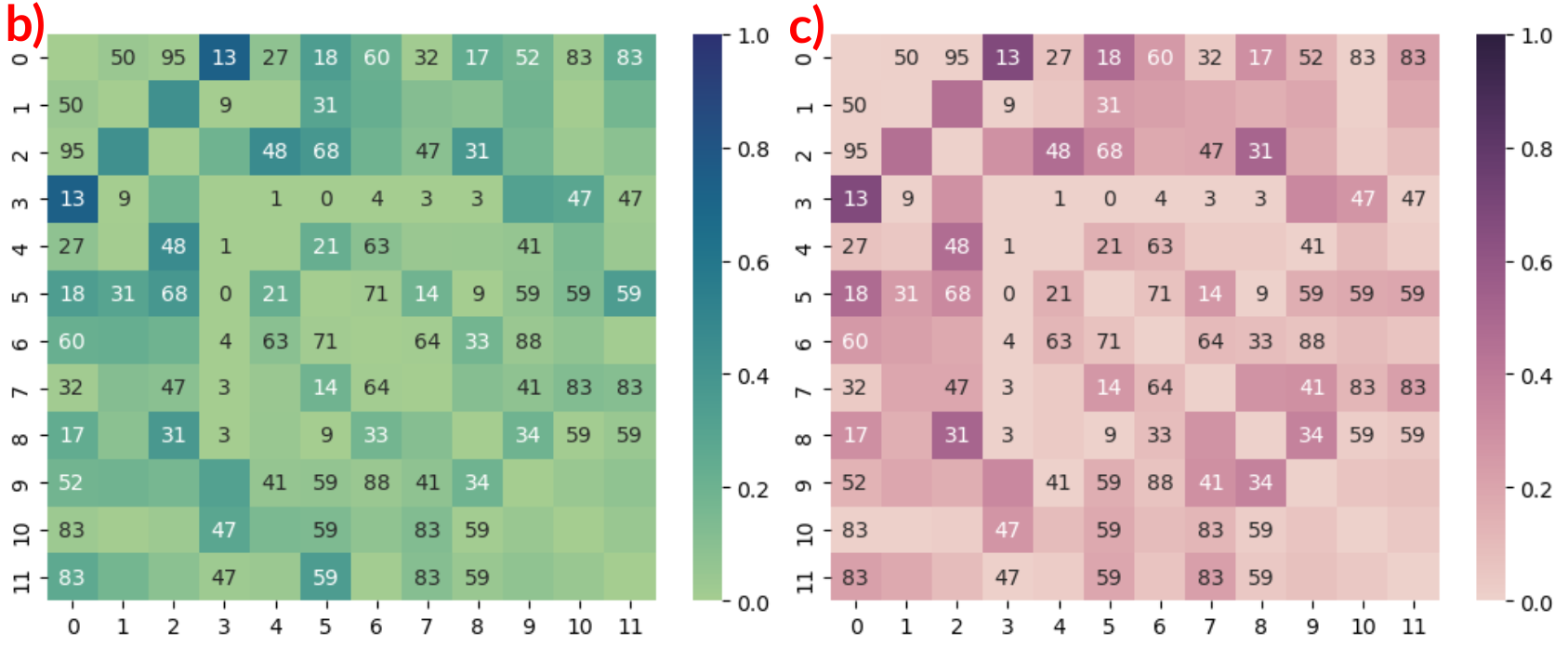}
\caption{Overview of Bivariate Comparison Between Synthetic and Real Data. Graph (a) shows pairwise scatter plots for pairs of PK target properties. Real data is marked in orange and synthetic data is marked in blue. The heatmap plots (b) and (c) are the Differential Pairwise Correlations (DPC) for pairs of PK target properties between real and synthetic data. The heatmap (b) graphs the DPC for the Pearson correlation coefficient. The heatmap (c) graphs the DPC for the Spearman correlation coefficient. PK target property values are numbered in order of (0) Caco2, (1) Lipophilicity, (2) AQSolDB, (3) FreeSolv, (4) PPBR, (5) VDss, (6) Half Life, (7) Clearance (Hep), (8) Clearance (Mic), (9) LD50, (10) hERG (1uM), and (11) hERG (10uM).}
\label{fig:realvssyn}
\end{figure*}
\subsection{Bivariate Distributions}

In addition to univariate comparisons, synthetic PK target properties can be compared to real data in terms of bivariate pairwise distributions and correlations. Bivariate pairwise scatterplots and Differential Pairwise Correlations (DPC) are shown in Figure \ref{fig:realvssyn}. Many pairwise combinations of PK target properties have very few overlapping real data values, and pairwise combinations with fewer than 100 examples have their cardinality numbered in the heatmaps in Figure \ref{fig:realvssyn}. We omit DPC values for pairwise combinations with cardinality less than 10. 

In combination with univariate HD, DPC provides a multivariate metric for evaluating the quality of synthetic data when compared to real data. We define the DPC as the absolute difference between the bivariate correlation coefficient of real and synthetic data as shown in Equation \ref{eq:cvcont}.

% Synthetic data that closely resembles real data should have similar bivariate pairwise correlations. In combination with the univariate HD metric, DPC provides a bivariate metric for open data policymakers to utilize as a standard to better compare synthetic datasets. If the real and synthetic datasets had high fidelity (i.e., the synthetic dataset closely resembled the real dataset), then the absolute difference would be close to 0 or very small. 

% For any fields containing continuous variables, the differential pairwise correlations in the real and synthetic data were evaluated to obtain fidelity in terms of bivariate statistics as shown in Equation \ref{cvcont}.
\begin{equation}\label{eq:cvcont}
\Delta CV_{{cont}_{XY}} = 
|\rho_{{XY}_{r}} - \rho_{{XY}_{s}}|
\end{equation}
where $X$ and $Y$ denote the two continuous variables, whereas $\rho_{XY}$ is the correlation coefficient for $X$ and $Y$. 
If the real and synthetic PK target property datasets are highly similar (i.e., the synthetic dataset closely resembles the real dataset), then the absolute difference would be close to 0 or very small. 
Heatmap (b) in Figure \ref{fig:realvssyn} shows DPC on the Pearson correlation coefficient (PCC). The average DPC for PCC is 0.123. Heatmap (c) in Figure \ref{fig:realvssyn} shows DPC on the Spearman correlation coefficient (SCC). cGAN and Syngand produce synthetic data with worse DPCs compared to Imagand. 
The average DPC for PCC for cGAN is 0.170, and 0.154 for Syngand, compared to 0.123 for Imagand. The average DPC for SCC for cGAN is 0.187, and 0.161 for Syngand, compared to 0.138 for Imagand.
% The average DPC for SCC is 0.138. 
These results indicate that the generated synthetic PK target properties resemble real data in pairwise correlations. 

Many pairwise combinations of the real data have a small cardinality of $<100$. As such, our synthetic PK target properties can benefit those pairwise combinations the most: researchers can augment pairwise real datasets with small cardinality to better answer pairwise target property research questions. Compared to pairwise target properties, overlap sparsity between combining multiple datasets results in even smaller cardinality. Scaling the S2PK model is straightforward, and can facilitate the generation of high-quality synthetic data that can be used to investigate multi-dataset research questions.
\begin{figure*}[h]
\centering
\includegraphics[width=1.8\columnwidth]{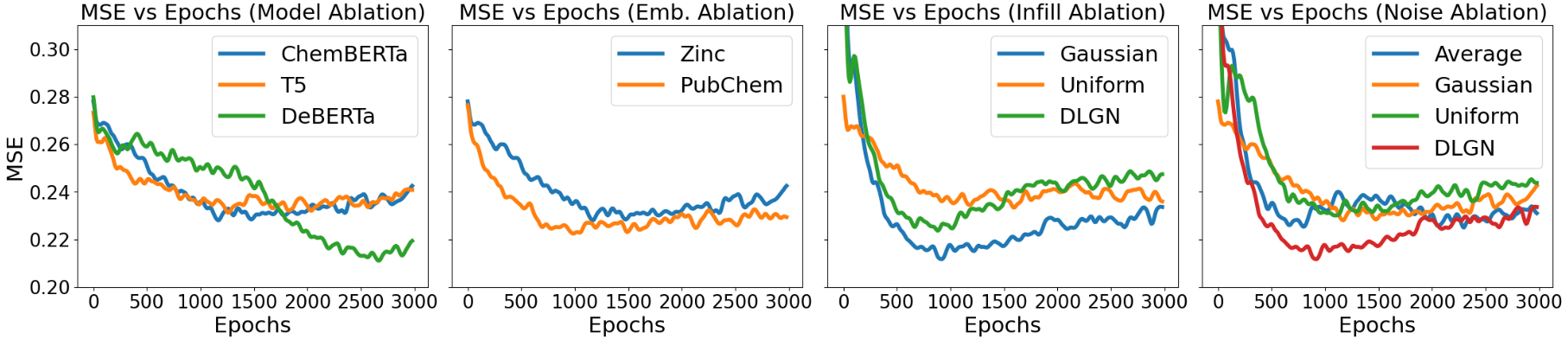} % Reduce the figure size so that it is slightly narrower than the column. Don't use precise values for figure width.This setup will avoid overfull boxes.
\caption{Mean Squared Error between real and synthetic target property data generated during training for different ablation experiments.}
\label{fig:valcurve}
\end{figure*}
\begin{table*}[t]
\centering
%\resizebox{.95\columnwidth}{!}{
\begin{tabular}{c c c c c c c c c c c c c c|c}
    \hline
    \textbf{Ablation} & \textbf{Exp} & \textbf{C2} & \textbf{Li.} & \textbf{Aq} & \textbf{FS} & \textbf{PP} & \textbf{VD} & \textbf{HL} & \textbf{C.(H)} & \textbf{C.(M)} & \textbf{LD50} & \textbf{h.1} & \textbf{h.10} & \textbf{Avg}\\
    \hline
    \multirow{3}{2em}{Pre.}
    & \multirow{1}{2em}{CBert} & 0.26 & 0.17 & 0.16 & 0.25 & 0.25 & 0.31 & 0.37 & 0.29 & 0.33 & 0.17 & 0.14 & \textbf{0.13} & 0.23 \\
    & \multirow{1}{2em}{DBert} & \textbf{0.21} & \textbf{0.16} & 0.18 & \textbf{0.20} & \textbf{0.22} & \textbf{0.27} & \textbf{0.36} & \textbf{0.24} & \textbf{0.28} & 0.17 & 0.15 & 0.15 & \textbf{0.21} \\
    & \multirow{1}{2em}{T5}    & 0.25 & 0.16 & \textbf{0.15} & 0.25 & 0.26 & 0.30 & 0.36 & 0.28 & 0.30 & \textbf{0.15} & \textbf{0.13} & 0.13 & 0.22 \\
    \hline
    \multirow{2}{2em}{Emb.} 
    & \multirow{1}{2em}{Zinc} & \textbf{0.26} & 0.17 & \textbf{0.16} & \textbf{0.25} & \textbf{0.25} & 0.31 & \textbf{0.37} & \textbf{0.29} & 0.33 & \textbf{0.17} & \textbf{0.14} & \textbf{0.13} & \textbf{0.23} \\
    & \multirow{1}{2em}{PubC} & 0.27 & \textbf{0.17} & 0.16 & 0.25 & 0.27 & \textbf{0.30} & 0.38 & 0.30 & \textbf{0.30} & 0.17 & 0.15 & 0.15 & 0.24 \\
    \hline
    \multirow{3}{2em}{Infill} 
    & \multirow{1}{2em}{Gaus} & \textbf{0.24} & \textbf{0.15} & \textbf{0.14} & \textbf{0.24} & \textbf{0.21} & \textbf{0.26} & \textbf{0.33} & \textbf{0.25} & 0.25 & \textbf{0.15} & \textbf{0.12} & \textbf{0.13} & \textbf{0.20} \\
    & \multirow{1}{2em}{Unif} & 0.28 & 0.19 & 0.18 & 0.26 & 0.27 & 0.31 & 0.38 & 0.30 & 0.30 & 0.19 & 0.14 & 0.14 & 0.25 \\
    & \multirow{1}{2em}{DLG}  & 0.26 & 0.16 & 0.15 & 0.25 & 0.23 & 0.29 & 0.36 & 0.27 & \textbf{0.24} & 0.16 & 0.13 & 0.13 & 0.22 \\
    \hline
    \multirow{4}{2em}{Noise} 
    & \multirow{1}{2em}{Avg}  & \textbf{0.26} & \textbf{0.15} & 0.22 & 0.28 & 0.25 & 0.39 & 0.39 & 0.28 & 0.33 & 0.18 & 0.15 & 0.16 & 0.25 \\
    & \multirow{1}{2em}{Gaus} & 0.26 & 0.17 & 0.16 & \textbf{0.25} & 0.25 & 0.31 & 0.37 & 0.29 & 0.33 & 0.17 & \textbf{0.14} & \textbf{0.13} & 0.23 \\
    & \multirow{1}{2em}{Unif} & 0.26 & 0.16 & 0.16 & 0.26 & 0.24 & 0.32 & 0.39 & 0.28 & 0.31 & \textbf{0.16} & 0.14 & 0.14 & 0.23 \\
    & \multirow{1}{2em}{DLG}  & 0.27 & 0.18 & \textbf{0.15} & 0.26 & \textbf{0.23} & \textbf{0.29} & \textbf{0.34} & \textbf{0.27} & \textbf{0.27} & 0.16 & 0.14 & 0.14 & \textbf{0.22} \\
    \hline
    \multicolumn{2}{r}{cGAN} &0.19&0.16&0.17&0.18&0.25&\underline{0.24}&\underline{0.28}&0.32&0.29&0.15&0.13&0.13&0.21\\
    \multicolumn{2}{r}{Syngand} &0.62&0.53&0.34&0.50&0.66&0.81&0.85&0.59&0.58&0.45&0.14&0.11&0.52\\
    \multicolumn{2}{r}{\textbf{Imagand (Ours)}} & \underline{0.19} & \underline{0.12} & \underline{0.13} & \underline{0.18} & \underline{0.20} & 0.27 & 0.36 & \underline{0.20} & \underline{0.19} & \underline{0.11} & \underline{0.09} & \underline{0.09} & \underline{0.18} \\
    \hline
\end{tabular}
\caption{Average Hellinger Distance Across 30 Generated Synthetic Target Property Datasets for Ablation Experiment Configurations. The best HD values for each ablation test are bolded. The best HD values across all ablation tests are underlined. HD values for our selected model configuration for MLE, univariate, and bivariate analysis are included in the table.}
\label{tab:ablation}
\end{table*}
% is used as it can be well-defined over continuous variables in the real and synthetic data.

% \begin{figure*}[h]
% \centering
% \includegraphics[width=1.8\columnwidth]{Figures/heatmap.png} % Reduce the figure size so that it is slightly narrower than the column. Don't use precise values for figure width.This setup will avoid overfull boxes.
% \caption{Model overview. Patches are generated from SMILES Embeddings combined with PK embeddings. The patches are then fed along with step embeddings into the base transformer model. A regression MLP head is used to produce the necessary output for denoising.}
% \label{fig:realvssyn}
% \end{figure*}

% Average Bivariate is 0.123 (PCC) and 0.138 (SCC).
\subsection{Ablation Studies}
We conduct ablation studies to investigate the performance of our S2PK model given different SMILES encoders, encoder training sets, and sampling approaches for the infilling and noise model. Ablation study results (Table \ref{tab:ablation}) are averages over 30 generated synthetic target property datasets, covering 90K target property values for ligands, for each ablation training run. From our ablation studies we find that Imagand generates more realistic synthetic data compared to cGAN and Syngand baselines in terms of univariate distributions.
Figure \ref{fig:valcurve} graphs MSE between real and synthetic data generated during training for ablation experiments. From our ablation studies, we motivate our selected model configuration. 

\subsubsection{Pre-trained SMILES Encoder} We select different pre-trained SMILES encoders and pretraining datasets for ablation. Among encoder models, DeBERTa performs the best in terms of average HD and synthetic and real data MSE. Among encoder training datasets, PubChem and Zinc have similar HD, with PubChem producing better synthetic and real data MSE. This motivates the choice of DeBERTa and PubChem for our selected model configuration.

\subsubsection{Discrete Local Gaussian Noise Model}
We select different infilling strategies and noise models for ablation.
Comparing noise model ablations, we measure the average MSE that each method injects into the data with Gaussian (1.19), uniform (0.53), and DLGN (0.29) ordered from most to least. This confirms that DLGN injects noise closely resembling the prior distribution. Similarly, we confirm DLGN has the best HD compared to Gaussian and uniform noise models. Comparing infilling ablations, DLGN has the best overall performance in HD and synthetic and real data MSE. This motivates the choice of DLGN for both infilling and noise models for our selected model configuration. 
%In future work, we will explore additional noise techniques, such as quantile transformations. 

\section{Discussions}
Our work is a major step towards building a new class of foundational models for drug discovery trained over a diverse range of datasets. Given the problem of data sparsity, Imagand can be utilized primarily as a \textit{in silico} pre-clinical tool, aimed to reduce the costs of \textit{in vitro} experiments and high-throughput screening. As a research tool, scientists can utilize our models to investigate and generate properties for novel molecules to be used for downstream PBPK simulations without costly assays. Even as an initial step, Imagand has many real-world pre-clinical applications where data sparsity and data scarcity are challenges.  

Although we cover a wide variety of ADMET datasets, most of these datasets are {\it in vitro}. One of the critical challenges in drug discovery is quantitative in vitro-to-in vivo extrapolation (QIVIVE). QIVIVE is an approach that extrapolates from in vitro concentration-response data to in vivo safe exposures or to identify exposure levels causing adverse effects. 
For future work, we will look to extend our model to include \textit{in vivo} datasets and to investigate new applications of Imagand for QIVIVE. 
Future work will look to further explore DLGN as an alternative to Gaussian noise, including implementing the derived formulation, conducting scaling law analysis, and to compare against a broad range of models and benchmarks. 
% We provide details on DLG sampling, the DLG diffusion derivation, and convergence analysis for DLG noise in the following appendix. Given our formulation of DLG, we show DDPM as a special case of DLG; we will further explore this connection in terms of the broader score-based generative modeling domain. 

\section{Conclusions}
The SMILES-to-Pharmacokinetic model Imagand generates synthetic PK target property data that closely resembles real data in univariate and bivariate distributions and for downstream tasks. Imagand provides a solution for the challenge of sparse overlapping PK target property data, allowing researchers to generate data to tackle complex research questions and for high-throughput screening. Future work will expand Imagand to categorical PK properties, and scale to more datasets and larger model sizes. 
In future work we will look to explore additional reparameterization tricks for diffusion, such as discrete diffusion \cite{austin2021structured}, to extend our methodology to be capable of learning and generating synthetic data following categorical and Log-logistic distributions common in drug discovery datasets.

\bibliography{example_paper}
\bibliographystyle{icml2025}

%%%%%%%%%%%%%%%%%%%%%%%%%%%%%%%%%%%%%%%%%%%%%%%%%%%%%%%%%%%%%%%%%%%%%%%%%%%%%%%
%%%%%%%%%%%%%%%%%%%%%%%%%%%%%%%%%%%%%%%%%%%%%%%%%%%%%%%%%%%%%%%%%%%%%%%%%%%%%%%
% APPENDIX
%%%%%%%%%%%%%%%%%%%%%%%%%%%%%%%%%%%%%%%%%%%%%%%%%%%%%%%%%%%%%%%%%%%%%%%%%%%%%%%
%%%%%%%%%%%%%%%%%%%%%%%%%%%%%%%%%%%%%%%%%%%%%%%%%%%%%%%%%%%%%%%%%%%%%%%%%%%%%%%
\newpage

\appendix
\onecolumn

\section{Discrete Local Gaussian Sampling}

% Discrete local gaussian (DLG) sampling is based on principles of inverse transform sampling. Inverse transform sampling is a method for sampling from any probability distribution given its cumulative distribution function. For any variable $X\in \mathbb{R}$, the random variable $F_X^{-1}(U)$ has the same distribution as $X$, where $F_X^{-1}$ is the generalized inverse of the cumulative distribution function $F_X$ of $X$ and $U \sim Unif[0,1]$. 

% As the true distribution is not always available, Discrete Local Gaussian Sampling $\mathbb{DLGS}$ looks to discretely approximate the cumulative distribution function by combining binning and Gaussian noise. Given $N$ bins, a discrete cumulative distribution function $\hat{F}_{X_N}^{-1}$ can be constructed. With $\phi$ as a scaling factor, we can then define DLG sampling as:

% \begin{align}
%     \mathbb{DLGS}(X=x) &= \mathcal{N}(x; \hat{F}_{X_N}^{-1}(U), \frac{\phi}{N^2}\mathbf{I})
% \end{align}

\begin{theorem}
\label{th:kl}
Let $D_{KL}(P||Q)$ be the Kullback-Leibler (KL) divergence for a model probability distribution $Q$ and the true probability distribution $P$. We then have:
\[lim_{N \to \infty} D_{KL}( X || \mathbb{\phi}(X=x)) = lim_{N \to \infty}\sum_\chi X log(\frac{X}{\mathbb{\phi}(X=x)}) = 0\]
\end{theorem}

\begin{proof}
\label{proof:kl}
Since $F_{X}^{-1} = lim_{N\to \infty}\hat{F}_{X_N}^{-1}$, and as the standard deviation $\frac{\sigma}{N} \to 0$ as $N\to\infty$, then $\mathbb{\phi}(X) \sim F_X^{-1}$, as $N\to\infty$. Since $F_X^{-1} = X$, then $D_{KL}( X || \mathbb{\phi}(X=x)) = 0$ as $N\to\infty$.
\end{proof}

From \autoref{th:kl} we see that DLG sampling becomes increasingly similar to the real data distribution as the number of bins increases. Empirically, using DLG sampling results in smoother training dynamics compared to Gaussian noise as well as higher-quality generated data. Theoretically, this may be because each de-noising step is smaller, which we see empirically, in turn making them easier to learn, when the noise is increasingly similar to the original distribution, especially given skewed or multi-modal real distributions. DLG noise also enables new training modalities, which we plan on exploring in future work; bypassing latent model training, and when prior real data distributions are known and well-defined.

\section{Training Details} \label{appdx:training}

\begin{table}[h]
    \centering
    \begin{tabular}{l r | l r}
        \hline
        \multicolumn{2}{c|}{Imagand Model} & \multicolumn{2}{c}{Diffusion Training} \\
        \hline
         Layers & 12 & Learning Rate & 1e-3 \\
         Heads & 16 & Weight Decay & 5e-2 \\
         MLP Dim. & 768 & Epoch & 3000 \\
         Emb. Dropout & 10\% & Batch Size & 256 \\
         Num Patches & 48 & Warmup & 200 \\
         Cond. Emb. Size & 768 & Timesteps (Train) & 2000 \\
         Time Emb. Size & 64 &  Timesteps (Infer.) & 150 \\
         PK Emb. Size & 256 & EMA Gamma ($\gamma$) & 0.994\\
         \hline
         
    \end{tabular}
    \caption{List of Imagand Model Hyperparameters used across experiments. Model hyperparameters include the number of layers, heads, multilayered perceptron (MLP) size, embedding dropout, and sizes for the conditional, time, and pharmacokinetic (Y) embeddings. Training hyperparameters include the learning rate, weight decay, number of epochs, batch size, warmup, diffusion timesteps used for training and inference, and the Exponential Moving Average (EMA) Gamma ($\gamma$).}
    \label{tab:hyperparameters}
\end{table}

We train a 19M parameter model for S2PK synthesis. 
% We use a batch size of 512 and 175K training steps for all models. 
Model hyperparameters were not optimized and are described in Table \ref{tab:hyperparameters}.
% Hyperparameters were not optimized with hyperparameter-tuning and will be covered in future work. 
% Figure \ref{fig:traincurve} shows the MSE loss curve over training for our selected model configuration and a few other ablation experiments. 
We do not find overfitting to be an issue. 
% Training for a few additional configurations might improve performance further. 
For classifier-free guidance, we joint-train unconditionally via dropout zeroing out sections of the SMILES embeddings with 10\% probability for all of our models. For the machine learning efficiency, and univariate and bivariate distribution analysis, we utilize DeBERTa embeddings trained on PubChem and DLGN for infilling and as the noise model. We compare our model configuration to other possible configurations in the ablation experiments. 
All experiments were conducted using a single NVIDIA GeForce RTX 3090 GPU. 

\subsection{Static Thresholding} We apply elementwise clipping the PK predictions to $[-1,1]$ as static thresholding, similar to \citet{saharia2022photorealistic, ddpm}. Since PK data is min-max scaled to the same $[-1,1]$ range as a preprocessing step, static thresholding is essential to prevent the generation of invalid and out-of-range PK values.  

\section{Convergence Analysis}\label{appdx:convergence}
Recent work by \citep{nakano2025} has proved the convergence of the original version of DDPM models as defined in \cite{ddpm}. As our work is based on the original version of DDPM with only a modification to the noise used, we leave the majority of the proof of convergence to the work by \citep{nakano2025}. Instead, to properly apply \citep{nakano2025} proof of convergence of DDPM to our formulation, we show how our modifications of DLG noise respect key conditions and lemmas required by \cite{nakano2025}. 
Relating to our noise modification is condition (H3) in \cite{nakano2025} proof of convergence
\begin{condition}[H3 in \cite{nakano2025}]
Let $z_i$ be the denoising term in a diffusion process. The function $z_i$ for the noise estimation satisfies
\[\underset{i=1,...,n}{max} ||z_i||_\infty = O((log\ n)^{\kappa_1}), \quad n \to \infty\]
for some constant $\kappa_1 > 0$
\end{condition}
\begin{proof}
    Utilizing conditions (H1), (H2), and Lemma 2 from \cite{nakano2025}, we have
    \[||\triangledown logp_i||_\infty \leq C_0/\sqrt{\bar\alpha_i}\leq C_0(log\ n)^{\kappa_1}\]
    Given the simplified version of the diffusion objective (as derived by \cite{ddpm}), it is known that the objective is equivalent to the score-matching objective \cite{chen2023samplingeasylearningscore}. More precisely
    \[ s_i(x) :=  - \frac{1}{\sqrt{1-\bar\alpha_i}} z_i(x)\]
    and the score function $\triangledown logp_i(\cdot)$ of $x_i$, $i=1,...,n$,
    \[\mathbb{E}|s_i(x_i) - \triangledown logp_i(x_i)|^2\]
    Hence, it is natural to assume that the norm of the estimated score function $s_i$ is bounded by $C^{\prime}_0(log\ n)^\kappa$ with some $C^{\prime}_0$. Given the following definition for DLG $z_i$
    \[z_i = \mathcal{N}(x; \hat{F}_{X_N}^{-1}(U), \frac{\phi}{N^2}\mathbf{I})\]
    where $z_i \sim F_X^{-1} = p_{data}$ leads to the condition (H3). 
\end{proof}

% Denote $p(t,x,r,y)$ the transition density of ${X_t}$, i.e.,

% \begin{align}
%     p(t,x,r,y) = \frac{1}{(2\pi \sigma^2_{t,r})^{d/2}}exp(-\frac{|y-m_{t,r}x|^2}{2\sigma_{t,r}^2}), \quad 0 < t < r, \ \  x,y \in \mathbb{R}^d
% \end{align}

% where $m_{t,r} = e^{-\frac{1}{2} \int^t_0 \beta_u du}$, $\beta = g^{\prime}$ where $g(t)$ is the linear interpolation of $\{0, -log\alpha_1,...,-\sum_i^nlog\alpha_i\}$ on $\{t_0, t_1,...,t_n\}$, and $\sigma_{t,r}=\sqrt{1-m^2_{t,r}}$. 

% The density function $p_t(y)$ of $X_t$ is given by

% \begin{align}
%     p_t(y) := \int_{\mathbb{R}^d} p(0,x,t,y)\mu_{data} (dx), \quad t>0, \ \  y \in \mathbb{R}^d
% \end{align}

% \begin{lemma}[Lemma 2. \cite{nakano2025}]
% The function  $\triangledown logp_t(x)$ is bounded and continuous on $[0,1] \times \mathbb{R}^d{}$ such that

% \[||\triangledown logp_t||_\infty \leq \frac{1}{m_{0,t}}||\triangledown logp_{data}||_\infty, \quad 0 \leq t \leq 1.\]
% \end{lemma}

% \begin{proof}
    
% \end{proof}

\section{Comparison to Baseline} \label{appdx:results}

We compare Imagand with a baseline in Conditional GAN (cGAN) \cite{mirza2014conditional} with 1.8M parameters and Syngand \cite{hu2024syntheticdatadiffusionmodels} with 9M parameters. Similar to in Imagand, SMILES-embeddings from a pre-trained T5 model are used conditionally by the cGAN model to generate PK properties as output for a specific drug. Compared to earlier results, \autoref{tab:cganmean}, \autoref{fig:synunivariate}, and \autoref{fig:cganunivariate} shows that Imagand is able to generate more realistic synthetic data compared to cGAN and Syngand.

% \begin{table*}[t]
% \centering
% %\resizebox{.95\columnwidth}{!}{
% \begin{tabular}{ccccccccccccc|c}
% \textbf{Model}&\textbf{C2}&\textbf{Li.}&\textbf{Aq}&\textbf{FS}&\textbf{PP}&\textbf{VD}&\textbf{HL}&\textbf{C.(H)}&\textbf{C.(M)}&\textbf{LD50}&\textbf{h.1}&\textbf{h.10}&\textbf{Avg}\\\hline
% \textbf{cGAN}&0.19&0.16&0.17&0.18&0.25&\underline{0.24}&\underline{0.28}&0.32&0.29&0.15&0.13&0.13&0.21\\
% \textbf{Syngand}&0.62&0.53&0.34&0.50&0.66&0.81&0.85&0.59&0.58&0.45&0.14&0.11&0.52\\
% \textbf{Our Model}&0.19&\underline{0.12}&\underline{0.13}&\underline{0.18}&\underline{0.20}&0.27&0.36&\underline{0.20}&\underline{0.19}&\underline{0.11}&\underline{0.09}&\underline{0.09}&\underline{0.18}\\\hline
% \end{tabular}
% \caption{Hellinger Distance Across Different Models. The best HD values across all models are underlined.}
% \label{tab:baselineablation}
% \end{table*}

\begin{table}[h]
\centering
\begin{tabular}{l|cccc|cccc}
\hline
\multicolumn{1}{c}{}&\multicolumn{4}{|c}{Mean}&\multicolumn{4}{|c}{Std}\\
Data&Real&Imgd&cGAN&Sygd&Real&Imgd&cGAN&Sygd\\\hline
\multirow{1}{*}{Caco2}&0.118&0.137&0.144&0.582&0.388&0.375&0.271&0.120\\
\multirow{1}{*}{Lipophilicity}&0.184&0.179&0.200&0.615&0.417&0.386&0.298&0.185\\
\multirow{1}{*}{AqSolDB}&0.106&0.107&0.132&0.102&0.412&0.362&0.291&0.182\\
\multirow{1}{*}{FreeSolv}&0.103&0.123&0.099&0.267&0.421&0.400&0.294&0.134\\
\multirow{1}{*}{PPBR}&0.570&0.562&0.628&0.974&0.496&0.467&0.367&0.091\\
\multirow{1}{*}{VDss}&-0.603&-0.615&-0.661&-0.985&0.442&0.389&0.310&0.065\\
\multirow{1}{*}{Half life}&-0.557&-0.559&-0.614&-0.982&0.450&0.419&0.315&0.058\\
\multirow{1}{*}{Clearance (H)}&-0.549&-0.559&-0.593&-0.974&0.605&0.551&0.465&0.086\\
\multirow{1}{*}{Clearance (M)}&-0.670&-0.676&-0.738&-0.985&0.445&0.382&0.312&0.064\\
\multirow{1}{*}{LD50}&-0.038&-0.054&-0.044&0.043&0.372&0.331&0.272&0.153\\
\multirow{1}{*}{hERG 1uM}&0.036&0.031&0.048&0.115&0.338&0.319&0.252&0.325\\
\multirow{1}{*}{hERG 10uM}&0.027&0.030&0.018&0.073&0.335&0.316&0.246&0.304\\\hline
\end{tabular}
\caption{Comparing mean and standard deviation values between real and synthetic target property values.} \label{tab:cganmean}
\end{table}

\begin{figure*}[h]
\centering
\includegraphics[width=0.5\columnwidth]{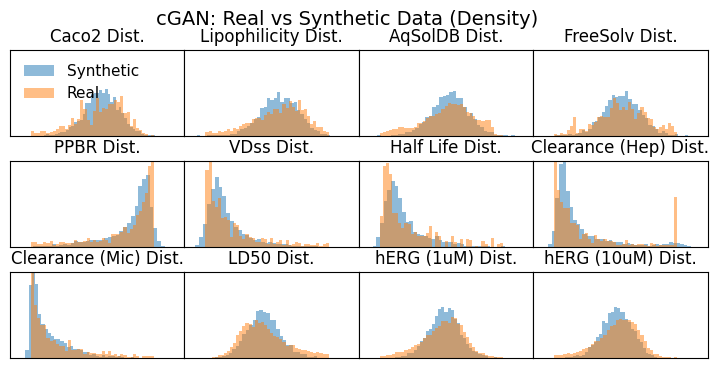} % Reduce the figure size so that it is slightly narrower than the column. Don't use precise values for figure width.This setup will avoid overfull boxes.
\includegraphics[width=0.3\columnwidth]{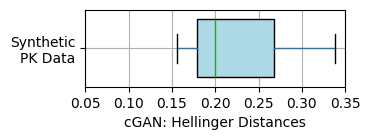}
\caption{Distributions of ligand PK properties and synthetic PK Data Hellinger Distances (HDs) for cGAN. Blue, synthetic distributions; orange, real distributions. 
% Log scale is used for the y-axis, in arbitrary unit.
}\label{fig:dist}
\label{fig:cganunivariate}
\end{figure*}

\begin{figure*}[h]
\centering
\includegraphics[width=0.5\columnwidth]{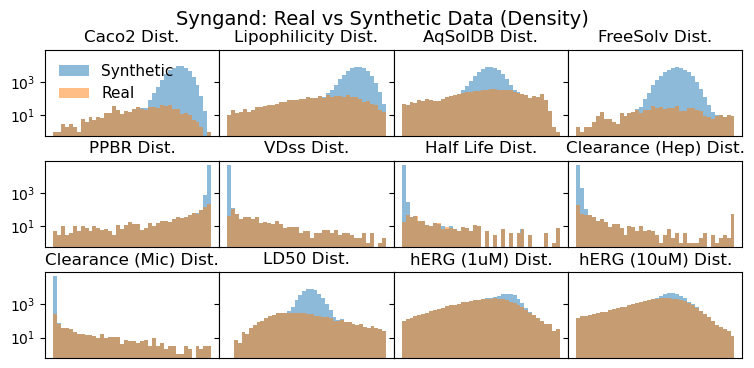} % Reduce the figure size so that it is slightly narrower than the column. Don't use precise values for figure width.This setup will avoid overfull boxes.
\includegraphics[width=0.3\columnwidth]{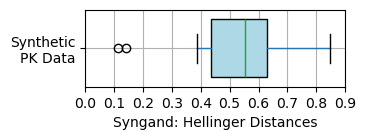}
\caption{Distributions of ligand PK properties (log-scale) and synthetic PK Data Hellinger Distances (HDs) for Syngand. Blue, synthetic distributions; orange, real distributions. 
% Log scale is used for the y-axis, in arbitrary unit.
}\label{fig:dist}
\label{fig:synunivariate}
\end{figure*}

\end{document}